\documentclass[12pt]{amsart}
\usepackage{amsmath, amssymb, amsthm, mathtools, abraces}
\usepackage{relsize, paralist, hyperref, cite, enumitem}
\usepackage[usenames,dvipsnames]{xcolor}
\usepackage{graphicx, subfigure, caption}
\usepackage{setspace}
\usepackage[paper=a4paper,left=3.2cm,right=3.2cm,top=3.1cm,bottom=3.1cm]{geometry}

\newtheorem{theorem}{Theorem}[section]

  \theoremstyle{definition}

  \theoremstyle{remark}

\numberwithin{equation}{section}

\newcommand{ \binomial }{ \operatorname{Binomial} }
\newcommand{ \mymod }[1]{ \;(\operatorname{mod}\, #1) }

\renewcommand{\qed}{\hfill $ \blacksquare $}

\begin{document}

\title[Maximum Entropy of a Sum of Discrete Random Variables]
      {On the Maximum Entropy of a Sum of Independent Discrete Random Variables}

\author{Mladen~Kova\v{c}evi\'{c}}


\thanks{The author is with the University of Novi Sad, Serbia. Email: kmladen@uns.ac.rs.}

\thanks{This work was supported by the European Union's Horizon 2020
research and innovation programme under Grant Agreement number 856967,
and by the Ministry of Education, Science and Technological Development
of the Republic of Serbia through the project number 451-03-68/2020-14/200156.}


\date{February 4, 2021.}


\begin{abstract}
Let $ X_1, \ldots, X_n $ be independent random variables taking values
in the alphabet $ \{0, 1, \ldots, r\} $, and $ S_n = \sum_{i = 1}^n X_i $.
The Shepp--Olkin theorem states that, in the binary case ($ r = 1 $),
the Shannon entropy of $ S_n $ is maximized when all the $ X_i $'s are
uniformly distributed, i.e., Bernoulli(1/2).
In an attempt to generalize this theorem to arbitrary finite alphabets,
we obtain a lower bound on the maximum entropy of $ S_n $ and prove
that it is tight in several special cases.
In addition to these special cases, an argument is presented supporting
the conjecture that the bound represents the optimal value for all $ n, r $,
i.e., that $ H(S_n) $ is maximized when $ X_1, \ldots, X_{n-1} $ are
uniformly distributed over $ \{0, r\} $, while the probability mass
function of $ X_n $ is a mixture (with explicitly defined non-zero weights)
of the uniform distributions over $ \{0, r\} $ and $ \{1, \ldots, r-1\} $.
\end{abstract}

\maketitle

\section{Introduction}
\label{sec:intro}

Maximum entropy probability distributions, being of interest in various
fields of science and engineering \cite{kapur}, have been studied extensively
in the literature.
Many of the canonical distributions from probability theory (e.g., uniform,
geometric, exponential, Gaussian) can be characterized as entropy maximizers
in natural families of probability laws.
Some more recent works \cite{harremoes, johnson, yu} have shown that Poisson
and binomial distributions are also maximum entropy distributions under
certain log-concavity constraints.
In this paper we consider the problem of entropy maximization for sums of
independent random variables, which has itself attracted a lot of interest
and is of importance in information theory in particular.
One of the most basic and well-known results in this area is the Shepp--Olkin
theorem \cite{mateev, shepp+olkin} which states that the entropy of a sum of
independent binary random variables is maximized when all the variables are
uniform, i.e., Bernoulli($ 1/2 $).
This statement has subsequently been strengthened in several respects (see,
e.g., the recent works \cite{hillion+johnson, hillion+johnson2, hillion+johnson3}
which settled the conjectures made in \cite{shepp+olkin}), and a continuous
version of the problem was analyzed in \cite{ordentlich, yu2}.
Virtually nothing is known about the problem for discrete variables over
non-binary alphabets.
In an attempt to generalize the Shepp--Olkin theorem to arbitrary finite
alphabets, we obtain a lower bound on the maximum entropy of a sum of
independent discrete random variables and prove that the bound is tight,
i.e., that it is in fact equal to the optimal value, in some particular
cases.

\subsection*{Notation, definitions, and auxiliary facts}

The Shannon entropy of a discrete random variable $ X $ with probability
mass function $ P_X $ supported on $ \{0, 1, \ldots, r\} $ is defined as
$ H(X) = H(P_X) = - \sum_{j=0}^r P_X(j) \log_2 P_X(j) $.
The following bounds on entropy are immediate from the definition:
$ 0 \leqslant H(X) \leqslant \log_2(r+1) $.
We shall also write $ h(p) = H(p, 1-p) $ for the binary entropy function.
The following elementary property of entropy will be useful in the analysis:
given a partition $ \{ A_1, \ldots, A_M \} $ of the alphabet (meaning that the
$ A_i $'s are pairwise disjoint and their union is $ \{0, 1, \ldots, r\} $),
we have
\begin{align}
\label{eq:decomposition}
  H(X)  =  \sum_{m=1}^{M}  \alpha_m H( X \,|\, X \in A_m )  +  H(\alpha_1, \ldots, \alpha_M) ,
\end{align}
where $ \alpha_m = P(X \in A_m) = \sum_{j \in A_m} P_X(j) $, and
$ H\!\left( X \,|\, X \in A_m \right) = H\!\left( P_{X | X \in A_m} \right) $
is the entropy of the conditional distribution $ P_{X | X \in A_m} $.

The binomial distribution with parameters $ n, p $ is denoted by
$ \binomial(n, p) $.
We shall use the symbol $ B_n $ for a generic random variable with
$ \binomial(n, 1/2) $ distribution.

A probability distribution (or any non-negative sequence) $ u_0, u_1, \ldots, u_n $
is said to be log-concave if $ u^2_{i} \geqslant u_{i-1} u_{i+1} $ for
all $ i = 1, \ldots, n-1 $.
It is said to be ultra-log-concave of order $ \infty $ if the stronger
condition $ i u^2_{i} \geqslant (i+1) u_{i-1} u_{i+1} $ holds, i.e., if
the sequence $ (u_i i!)_{i=0}^{n} $ is log-concave, and it is said to be
ultra-log-concave of order $ n $ if the still stronger condition
$ i (n-i) u^2_{i} \geqslant (i+1)(n-i+1) u_{i-1} u_{i+1} $ holds, i.e.,
if the sequence $ \left(u_i / \binom{n}{i} \right)_{i=0}^{n} $ is log-concave.

\section{The results}
\label{sec:results}

The following theorem presents a lower bound on the maximum value of
$ H(X_1 + \cdots + X_n) $ and claims that the bound is tight in the
case $ n = 2 $.

\begin{theorem}
\label{thm:main}
Let $ X_1, \ldots, X_n $ be independent random variables taking values
in $ \{0, 1, \ldots, r\} $, and let $ S_n = \sum_{i = 1}^n X_i $.
Then
\begin{align}
\label{eq:max}
  \max_{P_{X_1}, \ldots, P_{X_n}} H(S_n)  \;\geqslant\;
	  w_0 H(B_n) + (1 - w_0) \big( H(B_{n-1}) + \log_2(r-1) \big) + h(w_0)  ,
\end{align}
where
\begin{align}
\label{eq:w0}
  w_0 = \frac{ 2^{H(\!B_n\!) - H(\!B_{n-1}\!)} }{ r-1 + 2^{H(\!B_n\!) - H(\!B_{n-1}\!)} } .
\end{align}

For $ n = 2 $, equality holds in \eqref{eq:max} for all $ r \geqslant 1 $.
That is,
\begin{align}
\label{eq:maxn2}
  \max_{P_{X_1}, P_{X_2}} H(S_2) = 1 + \frac{w_0}{2} + (1-w_0) \log_2(r-1) + h(w_0) ,
\end{align}
where $ w_0 = \frac{ \sqrt{2} }{ r - 1 + \sqrt{2} } $.
\end{theorem}

Note that equality in \eqref{eq:max} holds also in the following cases:
\begin{itemize}[leftmargin=0.5cm]
\item
$ n = 1 $, $ r \geqslant 1 $.
In this case $ w_0 = \frac{2}{r+1} $ and the right-hand side of \eqref{eq:max}
reduces to $ \log_2(r+1) $ (it is understood that $ H(B_0) = 0 $).
It is well-known that $ \max_{P_{X}} H(X) = \log_2(r+1) $, and that the maximum
is attained when $ P_X $ is uniform.
\item
$ n \geqslant 1 $, $ r = 1 $.
In this case $ w_0 = 1 $ and the right-hand side of \eqref{eq:max}
reduces to $ H(B_n) $ (it is understood that $ (1 - w_0) \log_2(r - 1) = 0 $).
When $ r = 1 $, it is known that $ \max_{P_{X_1}, \ldots, P_{X_n}} H(S_n) = H(B_n) $,
and that the maximum is attained when the $ P_{X_i} $'s are all uniform.
This is precisely the statement of the Shepp--Olkin theorem.
\end{itemize}

\begin{proof}[Proof of Theorem \ref{thm:main}]
Select the following distributions for the random variables $ X_1, \ldots, X_n $:
$ Q_{X_i}(0) = Q_{X_i}(r) = \frac{1}{2} $ for $ i = 1, \ldots, n-1 $,
and $ Q_{X_n}(0) = Q_{X_n}(r) = \frac{w_0}{2} $,
$ Q_{X_n}(j) = \frac{1-w_0}{r-1} $ for $ j = 1, \ldots, r-1 $,
for some parameter $ w_0 $.
Let $ Q_{S_n} $ denote the corresponding distribution of $ S_n $.
With this choice of probability mass functions,
$ S_{n-1} = \sum_{i = 1}^{n-1} X_i $ is a Bernoulli sum having a
binomial distribution with parameters $ n-1 $ and $ 1/2 $
(and alphabet $ \{ k r : k = 0, 1, \ldots, n - 1 \} $), namely
\begin{align}
\label{eq:binom}
  Q_{S_{n-1}}(k r) = \binom{n-1}{k} 2^{-(n-1)} , \quad  k = 0, 1, \ldots, n-1 .
\end{align}
Consequently, for $ k = 0, 1, \ldots, n $ we have
\begin{align}
\label{eq:Sn}
	Q_{S_{n}}(kr)  =  Q_{S_{n-1}}(kr)  Q_{X_n}(0) + Q_{S_{n-1}}((k-1)r)  Q_{X_n}(r)
                 =  w_0 \binom{n}{k} 2^{-n}  ,
\end{align}
and, for $ k = 0, 1, \ldots, n -1 $ and $ j = 1, \ldots, r-1 $,
\begin{align}
\label{eq:Snj}
  Q_{S_{n}}(kr + j)  =  Q_{S_{n-1}}(kr)  Q_{X_n}(j)  
	                   =  \frac{1-w_0}{r-1} \binom{n-1}{k} 2^{-(n-1)}  .
\end{align}
Therefore, conditioned on the event $ S_n \equiv 0 \mymod{r} $, $ S_n \sim \binomial(n, 1/2) $,
and conditioned on $ S_n \equiv j \mymod{r} $, $ S_n \sim \binomial(n - 1, 1/2) $,
for every $ j = 1, \ldots, r-1 $.
In other words, the distribution of the random variable $ S_n $ is a
disjoint mixture of a $ \binomial(n, 1/2) $ distribution of weight $ w_0 $,
and $ r - 1 $ $ \binomial(n-1, 1/2) $ distributions of weight $ \frac{1-w_0}{r-1} $
each (see Figure \ref{fig:distribution}).
The entropy of this distribution is given precisely by the expression on
the right-hand side of \eqref{eq:max}.
The weight $ w_0 $ that maximizes this expression is the one in \eqref{eq:w0},
which can be shown directly by differentiating \eqref{eq:max}.

\begin{figure}
\centering
  \includegraphics[width=0.95\columnwidth]{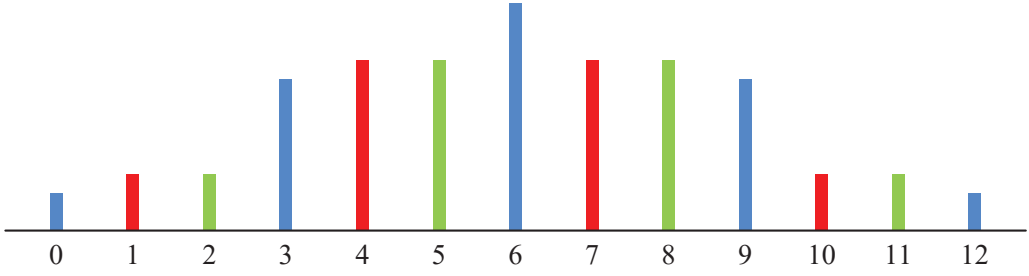}
\caption{The probability distribution of $ S_4 $ in the case when
$ P_{X_1} $, $ P_{X_2} $, and $ P_{X_3} $ are uniform over $ \{0, 3\} $,
and $ P_{X_4} $ is a mixture of the uniform distributions over $ \{0, 3\} $
and $ \{1, 2\} $.
For the purpose of illustration, the values $ P_{S_4}(a) $ are marked
in different colors for {\color{blue} $ {a \equiv 0 \mymod{3} } $},
{\color{red} $ a \equiv 1 \mymod{3} $}, and
{\color{LimeGreen} $ {a \equiv 2 \mymod{3} } $}.}
\label{fig:distribution}
\end{figure}%

To prove the second part of the statement, let $ n = 2 $, and let
$ P_{X_1}, P_{X_2} $ be generic distributions of $ X_1, X_2 $, and
$ P_{S_2} $ the induced distribution of $ S_2 $.
To simplify the notation, denote
$ w_0 = P(S_2 \equiv 0 \mymod{r}) = P_{S_{2}}(0) + P_{S_{2}}(r) + P_{S_{2}}(2r) $,
and $ w_j = P(S_2 \equiv j \mymod{r}) = P_{S_{2}}(j) + P_{S_{2}}(r+j) $
for $ j = 1, \ldots, r-1 $.
According to \eqref{eq:decomposition}, the entropy of $ S_2 $ can be decomposed
as follows:
\begin{align}
\label{eq:H2}
  H(S_2)  =  \sum_{j=0}^{r-1}  w_j H(S_2 \,|\, S_2 \equiv j \mymod{r})  +  H(w_0, w_1, \ldots, w_{r-1})  .
\end{align}
We will show that the choice of the probability mass functions from
the first part of the proof ($ Q_{X_1}, Q_{X_2} $) simultaneously maximizes
all the terms on the right-hand side of \eqref{eq:H2}, thereby maximizing
$ H(S_2) $ as well.
First, for $ j = 1, \ldots, r-1 $, it is clear that
$ H(S_2 \,|\, S_2 \equiv j \mymod{r}) \leqslant H(B_1) = 1 $, because the
conditional distribution $ P_{S_2 | S_2 \equiv j \mymod{r}} $ has only two
masses.
Now consider the case $ j = 0 $.
We have
\begin{align}
\nonumber
  P_{S_{2}}(0)   &=  P_{X_1}(0) P_{X_2}(0) ,  \\
\label{eq:S2}
	P_{S_{2}}(r)   &=  P_{X_1}(0) P_{X_2}(r) + P_{X_1}(r) P_{X_2}(0) + \sum_{j=1}^{r-1} P_{X_1}(j) P_{X_2}(r-j) ,  \\
\nonumber
	P_{S_{2}}(2r)  &=  P_{X_1}(r) P_{X_2}(r) .
\end{align}
It is straightforward to verify from these equations that the following holds:
\begin{align}
  P^2_{S_{2}}(r) - 4 P_{S_{2}}(0) P_{S_{2}}(2r)
	\geqslant  \big( P_{X_1}(0) P_{X_2}(r) - P_{X_1}(r) P_{X_2}(0) \big)^2
	\geqslant  0 ,
\end{align}
implying that the conditional distribution $ P_{S_2 | S_2 \equiv 0 \mymod{r}} $
is ultra-log-concave of order $ n = 2 $.
It was shown in \cite{yu} that the $ \binomial(n, 1/2) $ distribution is the
entropy maximizer in the class of all ultra-log-concave distributions of order $ n $,
and therefore $ H(S_2 \,|\, S_2 \equiv 0 \mymod{r}) \leqslant H(B_2) = \frac{3}{2} $.
From these observations and \eqref{eq:H2} we conclude that
\begin{align}
\label{eq:H22}
  H(S_2)  &\leqslant  w_0 H(B_2) + (1 - w_0) H(B_1)  +  H(w_0, w_1, \ldots, w_{r-1})  \\
\nonumber
          &\leqslant  w_0 H(B_2) + (1 - w_0) H(B_1)  +  H(w_0, 1-w_0) + (1-w_0) \log_2(r-1)  \\
\nonumber
					&=  w_0 H(B_2) + (1 - w_0) \big( H(B_{1}) + \log_2(r-1) \big) + h(w_0) ,
\end{align}
where the second inequality follows by partitioning $ \{0, 1, \ldots, r-1\} $
into $ \{0\} $ and $ \{1, \ldots, r-1\} $ and applying \eqref{eq:decomposition}
to the term $ H(w_0, w_1, \ldots, w_{r-1}) $.
Now \eqref{eq:max} and \eqref{eq:H22} imply \eqref{eq:maxn2}.
\end{proof}

Although we are at present in no position to prove such a statement,
it is tempting to conjecture that equality holds in \eqref{eq:max} for
all $ n, r $, i.e., that $ H(S_n) $ is maximized when $ X_1, \ldots, X_{n-1} $
are uniformly distributed over $ \{0, r\} $, while the probability
mass function of $ X_n $ is a mixture of the uniform distributions
over $ \{0, r\} $ and $ \{1, \ldots, r-1\} $.
In addition to the case $ n = 2 $, and several more special cases to
follow, where this is shown to be true, this claim would be in agreement
with the continuous version of the problem where it is known
\cite{ordentlich} that the differential entropy of a sum of independent
\emph{symmetric} random variables taking values in the interval $ [-1, +1] $
is maximized when $ X_1, \ldots, X_{n-1} $ are uniformly distributed,
i.e., Bernoulli($1/2$), on $ \{-1, +1\} $ and $ X_n $ is uniformly
distributed on the entire interval $ [-1, +1] $ (this is also conjectured
to be true without the symmetry assumption).

Before proceeding to the remaining special cases that we intend to
analyze, we reiterate once more the main idea behind Theorem \ref{thm:main}.
Our approach is to decompose the distribution $ P_{S_n} $ into $ r $
conditional distributions $ P_{S_n | S_n \equiv j \mymod{r}} $ of weight
$ w_j = P(S_n \equiv j \mymod{r}) $, $ j = 0, 1, \ldots, r-1 $, write
\begin{align}
\label{eq:max2}
  H(S_n)  =  \sum_{j=0}^{r-1}  w_j H(S_n \,|\, S_n \equiv j \mymod{r})  +  H(w_0, w_1, \ldots, w_{r-1})  ,
\end{align}
and then optimize the entropies and weights of each of these conditional
distributions.
These distributions are of course interdependent and it is not obvious
that they can be optimized separately.
However, guided by intuition, as well as by the problem's continuous
counterpart \cite{ordentlich}, one may ``guess'' that a (near) optimal
solution is obtained when the probability mass functions
$ P_{X_1}, \ldots, P_{X_{n-1}} $ are uniform on $ \{0, r\} $.
In this case the random variable $ S_n $, conditioned on the event
$ S_n \equiv j \mymod{r} $, has binomial distribution so we have
$ H(S_n \,|\, S_n \equiv 0 \mymod{r}) = H(B_n) $, and
$ H(S_n \,|\, S_n \equiv j \mymod{r}) = H(B_{n-1}) $ for $ j = 1, \ldots, r-1 $,
and the expression \eqref{eq:max2} reduces to
\begin{align}
\label{eq:max3}
  H(S_n)  =  w_0 H(B_n) + (1 - w_0) H(B_{n-1}) + H(w_0, w_1, \ldots, w_{r-1}) .
\end{align}
Moreover, in this case the weights of these conditional distributions
($ w_j $, $ j = 0, 1, \ldots, r-1 $) are controlled by the masses of the
$ n $'th random variable, $ X_n $.
Namely, $ w_0 = P(S_n \equiv 0 \mymod{r}) = P(X_n = 0) + P(X_n = r) $, and
$ w_j = P(S_n \equiv j \mymod{r}) = P(X_n = j) $ for $ j = 1, \ldots, r-1 $.
These weights can therefore be chosen separately in order to maximize the
expression in \eqref{eq:max3} and thus obtain a good lower bound, stated
in \eqref{eq:max}, on the maximum entropy of $ S_n $.
Further, by using the maximum entropy properties of the binomial distribution,
one may prove that the above choice of probability mass functions $ P_{X_i} $
is in fact optimal in some cases.
In particular, by the results of \cite{yu}, showing that the conditional
distributions $ P_{S_n | S_n \equiv j \mymod{r}} $ are ultra-log-concave
of order $ t $ (where $ t = n $ for $ j = 0 $ and $ t = n - 1 $ for
$ j = 1, \ldots, r-1 $) is sufficient to conclude that
$ H(S_n \,|\, S_n \equiv j \mymod{r}) \leqslant H(B_t) $ and that,
consequently, equality holds in \eqref{eq:max}.
This reasoning is used to establish the following claim as well.

\begin{theorem}
\label{thm:n3r2}
For $ n = 3 $, $ r = 2 $, equality holds in \eqref{eq:max}.
That is, over a ternary alphabet,
\begin{align}
\label{eq:n3r2}
  \max_{P_{X_1}, P_{X_2}, P_{X_3}} H(S_3)
	 =  \frac{3}{4} \big( 2 + (2 - \log_2 3) w_0 \big)  + h(w_0) ,
\end{align}
where $ w_0 = \frac{ (4/3)^{3/4} }{ 1 + (4/3)^{3/4} } $.
\end{theorem}
\begin{proof}
Let $ P_{X_1}, P_{X_2}, P_{X_3} $ be generic distributions over $ \{0, 1, 2\} $,
and $ P_{S_3} $ their convolution.
As per the above discussion, it suffices to prove that the conditional
distribution $ P_{S_3 | S_3 \equiv 0 \mymod{2}} $, resp.\
$ P_{S_3 | S_3 \equiv 1 \mymod{2}} $, is ultra-log-concave of order $ n = 3 $,
resp.\ $ n - 1 = 2 $.
Denote for brevity $ x_{a_1 a_2 a_3} = P_{X_1}(a_1) P_{X_2}(a_2) P_{X_3}(a_3) $,
and write
\begin{equation}
\label{eq:S3}
\begin{aligned}
  P_{S_3}(0) &= x_{000}  \\
  P_{S_3}(1) &= x_{100} + x_{010} + x_{001}  \\
	P_{S_3}(2) &= x_{200} + x_{020} + x_{002} + x_{110} + x_{101} + x_{011}  \\
	P_{S_3}(3) &= x_{210} + x_{201} + x_{021} + x_{120} + x_{012} + x_{102} + x_{111}  \\
	P_{S_3}(4) &= x_{220} + x_{202} + x_{022} + x_{211} + x_{121} + x_{112}  \\
	P_{S_3}(5) &= x_{221} + x_{212} + x_{122}  \\
	P_{S_3}(6) &= x_{222}  .
\end{aligned}
\end{equation}

In order to show that $ P_{S_3 | S_3 \equiv 0 \mymod{2}} $ is ultra-log-concave
of order $ 3 $, we need to demonstrate that
\begin{align}
\label{eq:ulc}
  \bigg( \frac{ P_{S_3}(2k) }{ \binom{3}{k} } \bigg)^2  \geqslant
	\frac{ P_{S_3}(2(k-1)) }{ \binom{3}{k-1} } \cdot \frac{ P_{S_3}(2(k+1)) }{ \binom{3}{k+1} }
\end{align}
for $ k = 1, 2 $.
Consider first the case $ k = 1 $.
Using the fact that each term in $ P_{S_3}(0) P_{S_3}(4) $ is equal to a term in
$ P^2_{S_3}(2) $, e.g.,
$ x_{000} x_{220} = x_{200} x_{020} $, $ x_{000} x_{211} = x_{200} x_{011} $, etc.,
one can verify that the quantity $ P^2_{S_3}(2) - 3 P_{S_3}(0) P_{S_3}(4) $ can
be represented as follows:
\begin{alignat}{2}
\label{eq:pom}
  &P^2_{S_3}&&(2) - 3 P_{S_3}(0) P_{S_3}(4)  \\
\nonumber
	  &= && \frac{1}{2} \big( x_{200} - x_{020} - x_{011} \big)^2
          + \frac{1}{2} \big( x_{020} - x_{002} - x_{101} \big)^2
          + \frac{1}{2} \big( x_{002} - x_{200} - x_{110} \big)^2  \\
\nonumber
		&\phantom{=}  && + \frac{1}{2} \big( x^2_{110} + x^2_{101} + x^2_{011} \big) + x_{200} x_{110} + x_{020} x_{011} + x_{002} x_{101}  \\
\nonumber
    &\phantom{=}  && + 2 \big( x_{200} x_{101} + x_{020} x_{110} + x_{002} x_{011} + x_{110} x_{101} + x_{110} x_{011} + x_{101} x_{011} \big) .
\end{alignat}
This expression is clearly non-negative, implying \eqref{eq:ulc}.
The proof for $ k = 2 $ is identical.

In order to prove that $ P_{S_3 | S_3 \equiv 1 \mymod{2}} $ is ultra-log-concave
of order $ 2 $, we need to establish the inequality
\begin{align}
\label{eq:ulc2}
  P^2_{S_3}(3)  \geqslant  4 P_{S_3}(1) P_{S_3}(5) .
\end{align}
Using the fact that each term in $ P_{S_3}(1) P_{S_3}(5) $ is equal to a term
in $ P^2_{S_3}(3) $, e.g., $ x_{100} x_{221} = x_{120} x_{201} $, one can obtain
the following identity:
\begin{alignat}{2}
\label{eq:pom2}
  &P^2_{S_3}&&(3) - 4 P_{S_3}(1) P_{S_3}(5)  \\
\nonumber
	  &= && \big( x_{210} - x_{012} + x_{201} - x_{021} + x_{120} - x_{102} \big)^2
     - 4 \big( x_{201} - x_{021} \big)\big( x_{120} - x_{102} \big)  \\
\nonumber
    &\phantom{=}  && + x_{111}^2 + 2 x_{111} \big( x_{210} + x_{201} + x_{021} + x_{120} + x_{012} + x_{102} \big) .
\end{alignat}
Now, if the terms $ x_{201} - x_{021} $ and $ x_{120} - x_{102} $ have
different signs, i.e., if $ (x_{201} - x_{021})( x_{120} - x_{102}) \leqslant 0 $,
then the expression in \eqref{eq:pom2} is certainly non-negative and
\eqref{eq:ulc2} follows.
On the other hand, if these two terms are both positive (resp.\ negative),
then the term $ x_{210} - x_{012} $ must also be positive (resp.\ negative).
To see this, write $ x_{201} > x_{021} $ and $ x_{120} > x_{102} $
(resp.\ $ x_{201} < x_{021} $ and $ x_{120} < x_{102} $), multiply the
corresponding sides of these inequalities to get
$ x_{201} x_{120} > x_{021} x_{102} $ (resp.\ $ x_{201} x_{120} < x_{021} x_{102} $),
and write out explicitly both sides of the latter inequality and cancel out
some of the common factors to conclude that it is equivalent to $ x_{210} > x_{012} $
(resp.\ $ x_{210} < x_{012} $).
The three terms $ x_{201} - x_{021} $, $ x_{120} - x_{102} $, $ x_{210} - x_{012} $
having the same sign implies that
\begin{align}
\label{eq:pom3}
  \big( &(x_{210} - x_{012}) + (x_{201} - x_{021}) + (x_{120} - x_{102}) \big)^2
	 - 4 \big( x_{201} - x_{021} \big)\big( x_{120} - x_{102} \big)  \\
\nonumber
    &>  \big( (x_{201} - x_{021}) + (x_{120} - x_{102}) \big)^2
		             - 4 \big( x_{201} - x_{021} \big)\big( x_{120} - x_{102} \big)  \\
\nonumber
    &=  \big( (x_{201} - x_{021}) - (x_{120} - x_{102}) \big)^2
		\geqslant  0 ,
\end{align}
which in turn implies that the expression in \eqref{eq:pom2} is positive,
i.e., that \eqref{eq:ulc2} holds.
The proof is complete.
\end{proof}

To conclude the paper, we state one more result in this direction that
generalizes both the Shepp--Olkin theorem and the above special cases.
Let $ X_1, \ldots, X_n $ be independent random variables, as before, but
now suppose that $ X_{1}, \ldots, X_{\ell} $ are taking values in $ \{0, 1, \ldots, r\} $,
while $ X_{\ell+1}, \ldots, X_{n} $ are taking values in $ \{0, r\} $.
Denote $ S_{\ell} = X_1 + \cdots + X_{\ell} $,
$ S'_{n-\ell} = X_{\ell+1} + \cdots + X_{n} $, and $ S_n = S_{\ell} + S'_{n-\ell} $.
Since $ S'_{n-\ell} $ is a Bernoulli sum taking values in
$ \{ k r : k = 0, 1, \ldots, n - \ell \} $, the conditional distribution
$ P_{S_{n} | S_{n} \equiv j \mymod{r}} $ is a convolution of
$ P_{S_{\ell} | S_{\ell} \equiv j \mymod{r}} $ and $ P_{S'_{n-\ell}} $,
for any fixed $ j = 0, 1, \ldots, r-1 $, meaning that
\begin{align}
\label{eq:Snl}
	P_{S_{n}}(mr+j)  =  \sum_{k=0}^{m} P_{S_{\ell}}((m-k)r+j) P_{S'_{n-\ell}}(kr) .
\end{align}
This implies%
\footnote{Convolution of an ultra-log-concave distribution of order $ m $ with
a Bernoulli distribution is itself ultra-log-concave of order $ m + 1 $ \cite[Lemma 1]{yu}.}
that $ P_{S_{n} | S_{n} \equiv j \mymod{r}} $ is ultra-log-concave
(of order $ n $ for $ j = 0 $, and order $ n - 1 $ for $ j = 1, \ldots, r-1 $)
whenever $ P_{S_{\ell} | S_{\ell} \equiv j \mymod{r}} $ is ultra-log-concave
(of order $ \ell $ for $ j = 0 $, and order $ \ell - 1 $ for $ j = 1, \ldots, r-1 $).
Together with Theorems \ref{thm:main} and \ref{thm:n3r2}, this proves the
following claim.

\begin{theorem}
\begin{inparaenum}
\item[(a)]
Let $ X_1, \ldots, X_n $, $ n \geqslant 2 $, be independent random variables,
and suppose that $ X_{1}, X_{2} $ are taking values in $ \{0, 1, \ldots, r\} $,
while $ X_{3}, \ldots, X_{n} $ are taking values in $ \{0, r\} $.
Then equality holds in \eqref{eq:max}.\\
\item[(b)]
Let $ X_1, \ldots, X_n $, $ n \geqslant 3 $, be independent random variables,
and suppose that $ X_{1}, X_{2}, X_{3} $ are taking values in $ \{0, 1, 2\} $,
while $ X_{4}, \ldots, X_{n} $ are taking values in $ \{0, 2\} $.
Then equality holds in \eqref{eq:max}.
\qed
\end{inparaenum}
\end{theorem}

\vspace{5mm}
\bibliographystyle{amsplain}

\vspace{1cm}

\end{document}